\documentclass[draft,12pt,reqno,a4paper]{amsart}

\usepackage[margin=3cm,footskip=1cm]{geometry}

\usepackage{amssymb} \usepackage{amsmath} \usepackage{amsfonts}
\usepackage{amsthm} \usepackage{mathtools}

\usepackage{enumerate} \usepackage[latin1]{inputenc}
\usepackage[shortlabels]{enumitem}

\usepackage{color}
\usepackage{graphicx} \usepackage{verbatim}

\DeclareMathOperator*{\slim}{s--lim}

\newcommand{\Ran}{{\operatorname{Ran}}}

\newcommand{\ad}{{\operatorname{ad}}}
\newcommand{\N}{{\mathbb{N}}} 
\newcommand{\R}{{\mathbb{R}}} 
\newcommand{\C}{{\mathbb{C}}}

 \renewcommand{\c}{{\rm c}}
\newcommand{\e}{{\rm e}} \newcommand{\ess}{{\rm ess}}
 \renewcommand{\i}{{\rm i}}
\renewcommand{\d}{{\rm d}}

\newcommand{\pupo}{{\rm pp}}

\renewcommand{\Re}{{\rm Re}\,} \renewcommand{\Im}{{\rm Im}\,}

\DeclarePairedDelimiter\inp\langle\rangle

\newcommand\parb[2][]{#1 \big ( #2#1\big )} \newcommand\parbb[2][]{#1
  \Big ( #2#1\Big )}

\newcommand{\mand}{\text{ and }}

 \newcommand{\bD}{{\bf D}}

 \newcommand{\bX}{{\bf X}}

\newcommand{\vA}{{\mathcal A}} \newcommand{\vB}{{\mathcal B}}

 \newcommand{\vH}{{\mathcal H}}

\newcommand{\vT}{{\mathcal T}} \newcommand{\vU}{{\mathcal U}}
\newcommand{\vV}{{\mathcal V}}

\newcommand\Step[1]{
  \par\bigskip
  \noindent
  \textit{#1}.\enspace
}

\newcommand\subStep[1]{
  \par\bigskip
  \noindent
  \underline{\textit{#1}}:\enspace
}

\theoremstyle{plain}
\newtheorem{thm}{Theorem}[section]
\newtheorem{proposition}[thm]{Proposition}
\newtheorem{lemma}[thm]{Lemma} 
\theoremstyle{definition}

 \newtheorem{cond}[thm]{Condition}

 \newtheorem*{remarks*}{Remarks}
\newtheorem*{remark*}{Remark}

\numberwithin{equation}{section}

\title[Rellich's theorem and $N$-body
systems]{Rellich's theorem and $N$-body
 Schr\"{o}dinger operators}

\thanks{
K.I. is supported by JSPS KAKENHI grant nr. 25800073.\, 
E.S. is supported by  DFF grant nr.  4181-00042.
}

\author{K. Ito}
\address[K. Ito]{Department of Mathematics, Kobe University\\
1-1 Rokkodai, Nada, Kobe, 657-8501, Japan}
\email{ito-ken@math.tsukuba.ac.jp}

\author{E. Skibsted} \address[E. Skibsted]{Institut for Matematiske
  Fag \\
  Aarhus Universitet\\ Ny Munkegade 8000 Aarhus C, Denmark}
\email{skibsted@imf.au.dk}

\begin{document}
\begin{abstract} We show an optimal version of Rellich's theorem  for
   generalized $N$-body  Schr\"{o}dinger operators. It applies to
   singular potentials, in particular to 
 a model  for   atoms and molecules with infinite mass and finite
  extent nuclei. Our proof relies on a  Mourre estimate \cite{IS} and
  a functional calculus  localization technique.
\end{abstract}

\maketitle

\section{Introduction  and results}\label{sec:introduction}
 Consider for given disjoint $R_1,\dots ,R_K\in \R^d$ the $N$-body Schr\"{o}dinger operator 
\begin{align}\label{eq:Hami_1M}
  H=\sum_{j = 1}^N \Big (-\frac{1}{2m_j}\Delta_{x_j}+\sum_{1\le k \le K}V^k_{j}(x_j-R_k)\Big ) + \sum_{1 \le i<j \le N} V_{ij}(x_i
  -x_j)
\end{align} describing  a system of $N$ $d$-dimensional particles in  $\Omega_1=\R^d\setminus
    \overline\Theta $,  where $\Theta=\cup_{1\le k \le K}
    \Theta_k$ for  given  open and bounded  (possibly empty) subsets  $\Theta_1,\dots, \Theta_K $
    of $
\mathbb{R}^d$. For those   (possibly existing) 
$k=1,\dots, K$  where $ \Theta_k\neq\emptyset$ we demand  that 
    $R_k\in  \Theta_k$, and for $N=1$ the last term to the right in
\eqref{eq:Hami_1M} is omitted. Furthermore we  assume that $d\geq2$ and that $V^{k}_{j}(y)=q_jq^k |y|^{-1}$
    and $V_{ij}(y)=q_iq_j|y|^{-1}$ (Coulomb potentials). We consider $H$ as an operator on the Hilbert space $\vH=L^2(\Omega)$,
$\Omega=(\Omega_1)^N$,  given  by imposing the
Dirichlet boundary condition.  In \cite{IS} we proved that the
    set of eigenvalues and thresholds (that is  eigenvalues for
    sub-Hamiltonians) is closed and countable and that
    any $L^2$-eigenfunction corresponding to a non-threshold eigenvalue
    is exponentially decaying. Under the additional condition that the
    exterior set $\Omega_1$  is 
    connected we proved that $H$ does not have positive
    eigenvalues.

In this paper we prove a version of Rellich's theorem
applicable  to this and other  generalized  $N$-body models.
 Letting $B_n=\{x\in \Omega|\,|x|<n\}$  we   introduce the Besov space $B_0^*$ of functions $\phi$ on $\Omega$ such that $1_{B_n}
\phi\in \vH$ for all $n\geq 1$  and 
\begin{subequations}
\begin{align}\label{eq:B0}
  \lim_{n\to \infty} n^{-1}\| 1_{B_n}\phi\|^2=0.
\end{align} Our version of Rellich's
theorem for the atomic  physics model \eqref{eq:Hami_1M}  reads:
\begin{proposition}\label{thm:Rehard} Suppose $H$ is given by 
  \eqref{eq:Hami_1M} with Coulomb potentials, and that $\phi$ is a  generalized
  eigenfunction  in $B_0^*$ fulfilling the Dirichlet boundary condition on $\Omega$ and corresponding to 
   a real
non-threshold eigenvalue. Then  $\phi\in\vH$ and therefore $\phi$ is a
genuine  eigenstate (in fact
$\e^{\epsilon |x|}\phi(x)\in L^2$ for some $\epsilon>0$).
  \end{proposition}

 We note that Isozaki in \cite{Is} showed for a class of smooth
 potentials in the
context of usual $N$-body operators (i.e. $N$-body operators without hard-core interaction)  that  any  non-threshold
generalized eigenfunction  in the weighted space 
$L^2_{\epsilon-1/2}=\inp{x}^{1/2-\epsilon}L^2$ must be in $L^2$ (here
the weight $\inp{x}=(1+|x|^2)^{1/2}$ and $\epsilon>0$). This  type
of 
problem  goes  back to Rellich \cite {Re}. The optimal one-body result is that any    non-threshold generalized eigenfunction  in
$B_0^*$ must be in $L^2$, cf. 
\cite{L1, L2, AH,Ho2}. The achievement  of this paper is three-fold. Whence we
obtain a version of Rellich's theorem for  $N$-body operators 1) with
singular pair potentials (with or without hard-core interaction), 2) stated with
the optimal space, and finally 3) with a proof that appears more
transparent (we think) than the one in \cite{Is}.

Introducing  the Besov space $B^*$ of functions $\phi$ on $\Omega$ such that $1_{B_n}
\phi\in \vH$ for all $n\geq 1$  and 
\begin{align}\label{eq:B}
  \sup_{n\geq 1 } n^{-1}\|1_{B_n}\phi\|^2<\infty,
\end{align}
\end{subequations}  we note for comparison that
$L^2_{\epsilon-1/2}\subseteq L_{-1/2}^2\subseteq
B_0^*\subseteq B^*\subseteq L_{-1/2-\epsilon}^2$.

The space $  B^*$ is a natural space for  generalized eigenfunctions. In
fact the generalized eigenfunctions in  the range of the delta function
\begin{align*}
  \delta(H-E)= \lim_{\epsilon\to 0}\,\tfrac 1\pi\Im(H-E-\i\epsilon)^{-1},
\end{align*}
 defined on yet another
Besov space $B\subseteq \vH$ 
 for any non-threshold and non-eigenvalue
$E$,  are  in $  B^*$ and they  span
the continuous subspace of $H$.  We shall not here elaborate on the
precise meaning of the above limit. For usual $N$-body operators an
optimal form of the limiting absorption principle was given in
\cite{JP}. For a version of LAP that applies to  more singular 
Hamiltonians (in particular  hard-core Hamiltonians) we refer to \cite[Appendix
B]{IS}. We claim, although here be unjustified, that the sharp form of
LAP is valid for the
Hamiltonians considered in \cite{IS} as well,  in particular for the one  of Proposition 
\ref{thm:Rehard}.  Whence for any non-threshold and non-eigenvalue
$E$ Proposition 
\ref{thm:Rehard} provides the following dichotomy for the mapping
$\delta(H-E):B\to B^*$. For $f\in B$ either
$\phi=\delta(H-E)f\in B^*\setminus B^*_0$ or $\phi=0$. This is a
general feature for $N$-body Hamiltonians  due to the present  paper. However, this
being said,  the generalized
eigenfunctions in  $B^*$ are poorly understood, see \cite{Ya} for some results. In particular it
is not known if $\Ran \parb{\delta(H-E)}$ span all generalized eigenfunctions in
$B^*$ (fulfilling the Dirichlet boundary condition for hard-core  Hamiltonians). It is believed  to be true and it is believed that all scattering information
is encoded in this space  (confirmed for the one-body
problem).

\subsection{Generalized $N$-body
  models}\label{subsec:Usual generalized $N$-body
  systems}  We will work in a generalized framework.  This is given by
a real finite dimensional vector space $\bX$ with an inner product   and a
finite family of subspaces $\{\bX_a| \ a\in \vA\}$ closed with respect
to intersection. We  refer to the elements of $\vA$ as {\it cluster
decompositions} although this terminology is not  motivated here. The orthogonal complement of $\bX_a$ in $\bX$ is
denoted $\bX^a$, and correspondingly we decompose $x=x^a\oplus x_a\in
\bX^a\oplus \bX_a$. We order $\vA$ by writing $a_1\subset a_2$ if
$\bX^{a_1}\subset \bX^{a_2}$. It is assumed that there exist
$a_{\min},a_{\max}\in \vA$ such that $\bX^{a_{\min}}=\{0\}$ and
$\bX^{a_{\max}}=\bX$. Let $\vB=\vA\setminus\{a_{\min}\}.$ 
  The length
of a chain of cluster decompositions $a_1\subsetneq \cdots \subsetneq
a_k$ is the number $k$. Such a  chain is said to connect $a=a_1$ and
$b=a_k$. The maximal length of all chains connecting a given $a\in
\vA\setminus\{a_{\max}\}$ and $a_{\max}$ is denoted by $\# a$. We
define $\# a_{\max}=1$ and denoting $\# a_{\min}=N+1$ we say the
family $\{\bX^a|a\in\vA\}$ is of {\it$N$-body type}.

To define  the generalized hard-core  model more structure is needed:
For each  $a\in \vB$ there is given an open subset
$\Omega_a\subset \bX^a$ with $\bX^a\setminus \Omega_a$ compact,
possibly $\Omega_a= \bX^a$. Let for $a_{\min}\neq b\subset a$
\begin{equation*}
  \Omega^a_b=\parb{ \Omega_b+\bX_b}\cap \bX^a=\Omega_b+\bX_b\cap \bX^a,
\end{equation*} and for $a \neq a_{\min}$ 
\begin{equation*}
  \Omega^a=\cap_{a_{\min}\neq b\subset a} \Omega^a_b.
\end{equation*} We define $\Omega^{a_{\min}}=\{0\}$ and $
\Omega=\Omega^{a_{\max}}$. 

Our ``soft potentials'' fulfill:

\begin{subequations}
\begin{cond}\label{cond:smooth2} 
  There exists $\varepsilon >0$ such that for all $b\in\vB$ there is a
  splitting into (real-valued) terms $V_b=V_b^{(1)}+V_b^{(2)}$, where
  \begin{enumerate}[label=(\arabic*)]
  \item \label{item:cond12b} $V_b^{(1)}$ is smooth on the closure of
    $\Omega_b$ and
    \begin{equation}
      \label{eq:1k2}
      \partial ^\alpha_yV_b^{(1)}(y)=O\big(|y|^{-\varepsilon-|\alpha|}\big ).
    \end{equation}
  \item \label{item:cond122} $V_b^{(2)}$ vanishes outside a bounded
    set in $\Omega_b$, and
    \begin{equation}
      \label{eq:2k2}
      V_{b}^{(2)}:H^1_0(\Omega_b)\to H^1_0(\Omega_b)^*\text{ is compact}.
    \end{equation}
  \end{enumerate}
\end{cond}  
\end{subequations}

We consider for $a\in\vB$ the Hamiltonian 
\begin{align*}
  H^a=-\tfrac 12
\Delta_{x^a} +V^a,\quad V^a(x^a)=\sum_{b\subset a}   V_{b}(x^{b}),
\end{align*}
 on the Hilbert space  $L^2(\Omega^a)$ with the Dirichlet boundary
condition on $\partial \Omega^a$, in particular
\begin{align*}
 H=H^{a_{\max}}=\tfrac 12 \Delta+V=H_0+V
\text{ on }\vH:=L^2(\Omega) 
\end{align*}
 with the Dirichlet boundary condition on $\partial
\Omega$.  More precisely the Hamiltonian $H^a$, henceforth referred to
as a
hard-core  Hamiltonian, is given by its quadratic form. The form  domain is the
standard Sobolev space $H^1_0(\Omega^a)$, and the corresponding action
is
the (naturally defined) 
Dirichlet form. Due to the continuous
embedding $H_0^1(\Omega^a)\subset H_0^1(\Omega^a_b)$ for
$a_{\min}\neq b\subset a$ we conclude that indeed $H^a$ is
self-adjoint, cf. \cite[Theorem X.17]{RS}. We define
$H^{a_{\min}}=0$ on $\C$. 
 The thresholds of $H$ are the  eigenvalues of
    sub-Hamiltonians, that is more precisely, the set of thresholds is
    given by 
\begin{align}
  \label{eq:thres}
  \vT = \cup_{a\in\vA, \# a\ge 2} \;\sigma_{\pupo}( H^a).
\end{align} This definition of thresholds conforms with the notion
of thresholds for usual $N$-body operators, and 
we recall that  the so-called HVZ theorem  asserts that the essential spectrum of $H$ is given by
the formula
\begin{equation*}
  \sigma_{\ess}(H)= [\min \vT,\infty),
\end{equation*}  cf. 
\cite[Theorem~XIII.17]{RS}. We proved in \cite{IS} that $\vT$ is closed and
countable. 
It is also known that under rather general conditions $H$ does
not have positive eigenvalues and the  negative non-threshold eigenvalues can at most
accumulate at the thresholds and only from below, cf. \cite{FH, IS, Pe}.

  In
\cite{IS} we proved a version of the Mourre
estimate using constructions of \cite{De, Gr}. This is more precisely
given in terms of  the (rescaled) Graf vector field 
 $\tilde \omega_R(x)=R\tilde \omega(\frac xR)$   and the corresponding
  operator
  \begin{align}\label{eq:oldA}
    A=A_R=\tfrac12\parb{\tilde \omega_R(x)\cdot p+p\cdot \tilde \omega_R(x)};\;R>1,\,p=-\i\nabla.
  \end{align} We    introduce a  function $d:\R\to\R$  by
  \begin{equation}\label{eq:44k}
    d(E)=\begin{cases}\inf _{\tau\in \vT (E)}(E-\tau),\;\vT
      (E):=\vT\cap \,]-\infty,E]\neq \emptyset,\\
      1,\;\vT
      (E)=\emptyset,
    \end{cases}
  \end{equation} where $\vT$ is given by \eqref{eq:thres}. 

  \begin{subequations}
\begin{lemma}\label{lemma:Mourre1_hard} For all $\kappa>0$ and compact 
    $I\subset\R$ with $I\cap\vT=\emptyset$ there
    exists $R_0>1$ such that for all $R\geq R_0$ and all $E\in I$:
    There exists  a
    neighbourhood $\vU$ of $E$ and a compact operator $K$ on
    $\vH$ such that
    \begin{equation}\label{eq:40kBB}
      f(H)^*\i [H,A_R]f(H)\geq f(H)^*\{2d(E)-\kappa
      -K\}f(H)\text { for all real }
      f\in C^\infty_{\c}(\vU).
    \end{equation}
  \end{lemma} The rescaled Graf vector field $\tilde \omega_R$ is complete on
$\Omega$. The Graf vector field is a gradient vector field, $\tilde
\omega=\nabla r^2/2$ for some positive function $r$. We also note that
by definition
the ``commutator''  $\i[H,A_R]$ is given by its formal expression
\begin{equation}\label{eq:mourre comm}
  \i [H,A_R]=p\tilde \omega_*(x/R)p-(2R^2)^{-1}\parb{\Delta^2
    r^2}(x/R)-\tilde \omega_R\cdot \nabla V.
\end{equation}    
  \end{subequations}

\subsection{Rellich's theorem in the generalized framework}\label{sebsec:result}  We need to be precise about the meaning of 
generalized eigenfunctions in $B_0^*$ fulfilling the Dirichlet
boundary condition: We say $\phi\in B_0^*$ (meaning that
\eqref{eq:B0} is fulfilled for the function $\phi: \Omega\to \C$) is a {\it generalized Dirichlet eigenfunction} with
eigenvalue  $E$ if for all $n\geq 1$ the function $\chi_n(|x|) \phi\in
H^1_0$ and in the distributional sense $(H-E)\phi$. Here $\chi\in
C^\infty(\R)$ is real-valued,
\begin{align}
\chi(t)
=\left\{\begin{array}{ll}
1 &\mbox{ for } t \le 1, \\
0 &\mbox{ for } t \ge 2,
\end{array}
\right.
\label{eq:14.1.7.23.24}
\end{align}  and $\chi_n(t)=\chi(t/n)$.  The main result 
of this paper reads
\begin{thm}\label{thm:priori-decay-b_0} Suppose Condition
  \ref{cond:smooth2}. Any generalized Dirichlet eigenfunction in $B_0^*$  with
   a real 
non-threshold eigenvalue  is in $\vH=L^2(\Omega)$.
\end{thm} We refer the reader to the proof of \cite[Corollary 1.8]{IS}
to see how Proposition \ref{thm:Rehard} follows from Theorem
\ref{thm:priori-decay-b_0}. Note that it follows from \cite[Appendix
C]{IS} that non-threshold $L^2$-eigenfunctions have exponential
decay. This is a consequence of Lemma \ref{lemma:Mourre1_hard}. Our 
proof of Theorem
\ref{thm:priori-decay-b_0} is also based on Lemma
\ref{lemma:Mourre1_hard}, however we shall proceed very
differently. Whence exponential
decay of  non-threshold  $ B_0^*$-eigenfunctions is obtained by a
combination of methods. A  uniform approach seems out of reach.

\section{Proof of Theorem
\ref{thm:priori-decay-b_0}}\label{sec:Proof}
We introduce and discuss various preliminaries  needed in the proof. It is easy to see that operators of the form $f(H)$, $f\in
C^\infty_\c(\R)$, preserve any of the weighted $L^2$-spaces and Besov
spaces introduced in Section \ref{sec:introduction}, cf. Lemma
\ref{lem:3.5} stated below and \cite[Theorem 14.1.4]{Ho2}. In particular $f(H)$ is a well-defined
operator on $B_0^*$. It is also easy to see that for any generalized
Dirichlet eigenfunction $\phi\in B_0^*$ with eigenvalue  $E$ we have
$f(H)\phi=f(E)\phi$. Whence if $f(E)=1$ we have the representation
$\phi= f(H)\phi$.

The proof of Theorem
  \ref{thm:priori-decay-b_0} will  rely on 
  Lemma \ref{lemma:Mourre1_hard}.   We
  abbreviate $A=A_R$ in Lemma \ref{lemma:Mourre1_hard}  using  the
  result  at
  a given non-threshold $E$ leading to  the
  following form  of the estimate
\begin{align}\label{eq:mour7d}
      f(H)\i [H,A]f(H)\geq f(H)\parb{\sigma-K}f(H),
    \end{align} where $\sigma>0$, $f\in C^\infty_\c(\R)$ is
    real-valued with   $f(E)=1$  and  $K$
    is compact.  Below the positivity of $\sigma$ will be crucial, but
    its size will not have importance.

Another ingredient of our proof is the operator
\begin{align*}
  B=\tfrac12\parb{\omega(x)\cdot p+p\cdot \omega(x)},
\end{align*} 
where $\omega=\omega_R=\tilde \omega_R/r_R$, $r_R(x)=Rr(\tfrac x R)$. Recall for comparison
that
\begin{align*}
  2A=\tilde\omega_R(x)\cdot p+p\cdot \tilde \omega_R(x)=r_R\omega_R\cdot p+p\cdot r_R\omega_R.
\end{align*} 
We shall suppress the dependence of the parameter $R$ (which is
considered  as a fixed large number). In particular we shall slightly abuse the
notation  writing for example $r$ rather than the rescaled
version $r_R$. Using the notation ${\bf D}$ for the Heisenberg
derivative $\i[H, \cdot ]$ we note the formal computations $B={\bf D}
r$, $A=r^{1/2}Br^{1/2}$ and
 \begin{align}\label{eq:formulasy}
 {\bf D}B=r^{-1/2}\parb{{\bD} A -B^2}r^{-1/2} +O(r^{-3}).
\end{align} 
Here the function
\begin{align*}
  O(r^{-3})=
\tfrac14 \omega\cdot (\nabla^2r)\omega/r^2=r^{-3}v(x),
\end{align*}
  where $v$ belongs to the algebra $\vV$ of smooth functions   on $\Omega$
obeying
\begin{equation*}
    \forall \alpha\in \N_0^{\dim \bX}\mand
    k\in \N_0: |\partial_x^\alpha(x\cdot
    \nabla)^k v(x)|\leq C_{\alpha,k}.
  \end{equation*}  This is due to the fact that  the function
  $r^2-x^2\in 
  \vV$, cf. \cite[(2.1b)]{IS}.  We note that
  the  expression \eqref{eq:mourre comm} for ${\bD} A$  takes
  the form 
  \begin{align*}
    {\bD} A=\sum_{|\alpha\leq 2}v_\alpha p^\alpha;\,v_\alpha\in \vV,
  \end{align*} which make sense as a bounded form on
  $\vH^1=Q(H)=H^1_0$. For computations we will need the following
  rigorous version
  of \eqref{eq:formulasy}, cf. \cite[Lemma A.8]{IS}: In the sense of
  strong limit in the space of bounded operators $\vB(\vH^1, \vH^{-1})$
  \begin{align}
    \label{eq:limit}
     {\bf D}B:=\slim_{t\to 0} t^{-1}\parb{H\e^{\i tB}-\e^{\i tB}H}=r^{-1/2}\parb{{\bD} A -B^2}r^{-1/2} +r^{-2}v,
  \end{align} where $v \in\vV$. Here we use that $\e^{\i
    tB}\in\vB(\vH^1)\cap\vB(\vH^{-1})$, cf. \cite[Lemma A.6]{IS}. With a similar
  interpretation it follows that $\ad_B({\bf D}B)=[{\bf D}B,B]\in\vB(\vH^1, \vH^{-1})$.
   Although this will not be needed in fact  all higher order repeated commutators $\ad^k_B({\bf D}B)$
  exist in this sense.

The  above computations  were used in \cite{GIS}, in fact our 
  proof of  
  Theorem \ref{thm:priori-decay-b_0} will to a large extent rely  on
  ideas from \cite{GIS}  similarly  to
  \cite{Is}.  On the more technical level a certain part of Isozaki's
  proof  also relies on ideas of \cite{FH},  whereas the analogous
  difficulty at the threshold  zero for a one-body problem was treated in 
  \cite{Sk} using a propagation of singularity result. We will present a new
 approach  based on a conveniently chosen partition
  of unity.  

  As in \cite{GIS} phase-space localization in terms of functional
  calculus of the operators $B$ and $H$ will be important. If $f\in
  C^\infty_\c(\R)$ (or for example $f=\chi$) the operator $f(B)$
  preserves the spaces mentioned at the beginning of this section  (like $f(H)$ does), and more
  generally $f(B)$ and $f(H)$ enjoy good properties regarding
  commutation with functions of $r$. These properties are studied in
  \cite[Section 2]{GIS} and will be used frequently below, however
  our presentation will be self-contained. It is   based on an 
  abstract result
  in Appendix
  \ref{sec:Functional calculus}.

We need a  cut-off function $\eta_\epsilon\in C^\infty(\R)$ with
special 
properties: The parameter $\epsilon>0$ is considered small, and we
define $\eta_\epsilon(x)=\tfrac 1\epsilon\eta(\tfrac x\epsilon)$, where $\eta'(x)> 0$ for
$|x|<1$,  $\eta(x)=0$ for $x\leq -1$ and $\eta(x)=1$ for $x\geq 1$. We can
construct $\eta$ such  that $\eta'$ is even, $\sqrt{\eta'}\in C^\infty(\R)$ and   for some 
$c>0$ 
\begin{align}
  \label{eq:constr}
  \eta'(x)\geq c\,\eta(x)\text{ for }x\in(-1,1/2].
\end{align} The optimal  choice of such $c$ (a
necessary condition is $c\leq 2\ln 2$) is not important for us
since we will only need \eqref{eq:constr}  in the following disguised
form: For any $\tilde c >0$ and all $\epsilon$ small enough ($\epsilon^2\leq \tfrac 23 c\tilde c$ suffices)
\begin{align}
  \label{eq:fundest}
  (\tfrac \epsilon2-b)\eta_\epsilon(b)\leq \tilde c
  \,\eta'_\epsilon(b)\text{ for all }b\in \R.
\end{align}
Note also that since  $\eta_\epsilon'$ is even
$\epsilon^{-1}=\eta_\epsilon(b)+\eta_\epsilon(-b)$. We shall use the function  $\zeta_\epsilon(b)=\eta_\epsilon(b)-\eta_\epsilon(-b)$.

Let $h_\kappa(r)=\tfrac r {1+\kappa r}$ for
    $\kappa\geq 0 $,  and let $X_\kappa$  and $Y_\kappa$  be the operators  of
multiplication by $h_\kappa$ and $\tfrac 1 {1+\kappa r}$,
respectively. Writing    $X=X_0$ we note that $X_\kappa=X
Y_\kappa$. Note also that $\nabla h_\kappa(r)= (1+\kappa r)^{-2}\omega$,
whence for example $\
\bD {X_\kappa}={Y_\kappa}B{Y_\kappa}$.

\begin{proof}[Proof of Theorem \ref{thm:priori-decay-b_0}]
  Let  $\phi\in B_0^*$ be  a  given generalized Dirichlet
  eigenfunction  with a non-threshold eigenvalue  $E$. We    shall first  show
  that $\phi\in L^2_{-1/4}$ and then use this property to show by a
  similar procedure  
  that $\phi\in L^2$. Write $\phi=f(H)\phi$ where $f\in
  C^\infty_\c(\R)$, $f(E)=1$ and \eqref{eq:mour7d} holds.

\Step {Step I} We outline the ideas of the proof of the property
$\phi\in L^2_{-1/4}$. We shall consider the ``propagation observable''
    \begin{align*}
      P=P_{\kappa}=f(H) X_\kappa^{1/4}\zeta_\epsilon(B)X_\kappa^{1/4}f(H),\,\kappa>0.
\end{align*} The parameter   $\epsilon>0$ will be 
 fixed shortly   small enough. Note that $X_\kappa$ and  $P_\kappa$ are  bounded due to the appearance of the
factor $\tfrac 1 {1+\kappa r}$. Eventually this factor will  be removed by
letting $\kappa\to 0$.
 More precisely we  shall demonstrate some ``essential  positivity''
 of $\bD P$ persisting in
 the
 $\kappa\to 0$ limit. For any  $n$ the function  $\phi_n=\chi_n(r)\phi\in H^1_0$, $(H-E)\phi_n=-\i
 (\bD \chi_n)\phi$ and whence the expectation 
 \begin{subequations}
 \begin{align}\label{eq:virial}
   \inp{\bD P}_{\phi_n}=-2\Re\inp{( \bD \chi_n)P\chi_n}_\phi.
 \end{align}
 Since $\phi\in B_0^*$ the term  to the right vanishes as $n\to
 \infty$, so it remains to study  the  term to the left in this limit. We compute
 \begin{align*}
   4\bD
   X_\kappa^{1/4}={Y^2_\kappa}X_\kappa^{-3/4}B+\tfrac3{8}
\i \omega^2 X_\kappa^{-7/4} Y_\kappa^4+
\i \kappa\omega^2 X_\kappa^{-3/4} Y_\kappa^3.
 \end{align*} With  commutation errors this should give
 \begin{align}\label{eq:com1}
   \begin{split}
   &2\Re\parbb{\parb{\bD
   X_\kappa^{1/4}}\zeta_\epsilon(B)X_\kappa^{1/4}}\\&= 
      \tfrac12Y_{\kappa}X_\kappa^{-1/4}B\zeta_\epsilon(B)X_\kappa^{-1/4}Y_{\kappa}+
   X^{-3/4}O(\kappa^{0})X^{-3/4}.  
   \end{split}
 \end{align} 

Similarly, letting $\theta_\epsilon=\sqrt{\eta'_\epsilon}$,
\begin{align}\label{eq:com2}
\begin{split}
  X_{\kappa}^{1/4}(\bD \zeta_\epsilon(B))X_{\kappa}^{1/4}
\approx X_{\kappa}^{1/4}\parb{\theta_\epsilon(B)(\bD B)\theta_\epsilon(B)+\theta_\epsilon(-B)(\bD B)\theta_\epsilon(-B)}X_{\kappa}^{1/4}.
\end{split}
  \end{align} Now we insert \eqref{eq:formulasy}
  into this expression  and then estimate by
  \eqref{eq:mour7d} (here ignoring  factors of $f(H)$). Ignoring
   the contribution from the compact
  operator $K$ in \eqref{eq:mour7d} (which  should be controllable in the state
  $\phi_n$)  we should 
  end up with the effective lower bounds
  \begin{align}\label{eq:lowbnd0}
\begin{split}
    &X_{\kappa}^{1/4}(\bD
    \zeta_\epsilon(B))X_{\kappa}^{1/4}\\&\geq(\sigma-\epsilon^2)X_\kappa^{1/4}X^{-1/2}\parbb{\eta'_\epsilon(B)+\eta'_\epsilon(-B)}X^{-1/2}X_\kappa^{1/4} +
   X^{-3/4}O(\kappa^{0})X^{-3/4}\\&\geq (\sigma-\epsilon^2) Y_{\kappa}X_\kappa^{-1/4}\parbb{\eta'_\epsilon(B)+\eta'_\epsilon(-B)}X_\kappa^{-1/4}Y_{\kappa}+
   X^{-3/4}O(\kappa^{0})X^{-3/4}.
\end{split}
  \end{align} This suggests we should have    $\epsilon$ so small  that
  $2\epsilon^2< \sigma$ since then
   $\sigma-\epsilon^2\geq \sigma/2$. In total we  are  lead to
 consider
 \begin{align*}
   \tfrac12Y_{\kappa}X_\kappa^{-1/4}\parb{B\eta_\epsilon(B)-B\eta_\epsilon(-B)+\sigma\eta'_\epsilon(B)+\sigma\eta'_\epsilon(-B)}X_\kappa^{-1/4}Y_{\kappa}.
 \end{align*} Using \eqref{eq:fundest} with $\tilde c=\sigma$ and 
  any  possibly smaller 
 $\epsilon$, whenceforth considered fixed,  we arrive  at the lower bound
\begin{align*}
 \tfrac{\epsilon} 4Y_{\kappa}X_\kappa^{-1/4}\parb{\eta_\epsilon(B)+\eta_\epsilon(-B)}{X_\kappa}^{-1/4}Y_{\kappa}=\tfrac{1} 4X_\kappa^{-1/2}Y^2_{\kappa}.
\end{align*} 
\end{subequations}

Whence,  given that    error terms can  be controlled, we  obtain from these
arguments  the uniform bound
\begin{align}\label{eq:fundests}
  \|X_\kappa^{-1/4}Y_{\kappa}\phi\|^2=\lim_{n\to \infty}\|X_\kappa^{-1/4}Y_{\kappa}\phi_{n}\|^2\leq C_\phi.
\end{align} 
By letting $\kappa \to 0$ in \eqref{eq:fundests} it follows that
$\phi\in L^2_{-1/4}$. The constant $C_\phi$ arises from bounding
errors in \eqref{eq:com1}--\eqref{eq:lowbnd0} as well as from bounding
the contribution from the operator $K$  in \eqref{eq:mour7d} to supplement \eqref{eq:lowbnd0}. As we will
see, in agreement with \eqref{eq:com1} and \eqref{eq:lowbnd0}  it can be
taken of the form $C_\phi=C\|X^{-3/4}\phi\|^2$.

 \Step{Step II} To implement the scheme of Step I  we provide
details on estimating errors from various commutation and the outlined
application of \eqref{eq:mour7d}.

\subStep{Right-hand side  of \eqref{eq:virial}}
We have
\begin{align*}
  \bD \chi_n=\chi'_nB-
\tfrac {\i}2\omega^2\chi''_n,
\end{align*} and therefore using that $BP$ is bounded and that
$\phi\in B_0^*$ indeed
\begin{align*}
  \lim_{n\to \infty}\Re\inp{( \bD \chi_n)P\chi_n}_\phi=0.
\end{align*}

\subStep{\it Left-hand side of \eqref{eq:virial}} We calculate
\begin{align*}
  \bD P=\i f(H)\big[\tilde H,X_\kappa^{1/4}\zeta_\epsilon(B)X_\kappa^{1/4}\big]f(H),
    \end{align*} where $\tilde H=g(H)$ with $g(\lambda)=\lambda
    \tilde f(\lambda)$ and  real-valued $\tilde f\in
    C_\c^\infty(\R)$  chosen such that  $\tilde f(\lambda) =1$ in a neighbourhood of
    the support of $f$. Whence  denoting   ${\tilde \bD}$  the
    Heisenberg derivative $\i[\tilde H, \cdot ]$  we need to
    examine \eqref{eq:com1} and \eqref{eq:com2} with ${ \bD}$
    replaced by ${\tilde \bD}$.

\subStep{\it \eqref{eq:com1} with ${\tilde \bD}$} Based on the
representation \eqref{83a} (applied to $B=H$) we compute 
\begin{align*}
   4\tilde \bD
   {X_\kappa}^{1/4}=g'(H)Y^2_\kappa X_\kappa^{-3/4}B+X^{-3/4}O(\kappa^{0})X^{-1},
 \end{align*} and therefore 
\begin{align*}
  4f(H)\parb{\tilde \bD
   {X_\kappa}^{1/4}}=f(H)Y^2_\kappa X_\kappa^{-3/4}B+X^{-3/4}O(\kappa^{0})X^{-1}.
\end{align*} Using again Lemma \ref{lem:3.5} we then obtain 
\begin{align*}
   \begin{split}
   &2f(H)\Re\parbb{\parb{\tilde \bD
   X_\kappa^{1/4}}\zeta_\epsilon(B)X_\kappa^{1/4}}f(H)\\&=
      \tfrac12f(H)Y_{\kappa}{X_\kappa}^{-1/4}B\zeta_\epsilon(B){X_\kappa}^{-1/4}Y_{\kappa}f(H)+
   X^{-3/4}O(\kappa^{0})X^{-3/4}.  
   \end{split}
 \end{align*}

 \subStep{\it \eqref{eq:com2} and \eqref{eq:lowbnd0} modified} Using
 \eqref{eq:limit} and Lemma \ref{lem:3.5}  we can write
 \begin{align}\label{eq:dif}
\begin{split}
  X_{\kappa}^{1/4}\parb{\tilde \bD \zeta_\epsilon(B)}X_{\kappa}^{1/4}
&= X_{\kappa}^{1/4}\parb{\theta_\epsilon(B)(\tilde \bD
  B)\theta_\epsilon(B)+\theta_\epsilon(-B)(\tilde \bD B)\theta_\epsilon(-B)}X_{\kappa}^{1/4}\\&
   +X^{-3/4}O(\kappa^{0})X^{-3/4}.
\end{split}
\end{align} Indeed by \eqref{83a}  (with interchanged roles of
 $B$ and $H$) and \eqref{eq:limit} 
  \begin{align*}
   \tilde \bD
  B=\int_{\C} (H-z)^{-1} (\bD
  B)(H-z)^{-1}\d \mu(z)\in \vB(\vH),
  \end{align*} and therefore in turn similarly 
\begin{align*}
   [\tilde \bD
  B,B]&=\slim_{t\to 0} \tfrac 1{\i t}\parbb{\parb{\tilde \bD B}\e^{\i tB}-\e^{\i
      tB}\parb{\tilde \bD B}}\\&=\int_{\C} (H-z)^{-1} [\bD
  B,B](H-z)^{-1}\d \mu(z)\\&+2\i\int_{\C} (H-z)^{-1}(\bD
  B)(H-z)^{-1} (\bD
  B)(H-z)^{-1}\d \mu(z)\in \vB(\vH).
  \end{align*} In fact it follows this representation that 
\begin{align*}
   X^s[\tilde \bD
  B,B]X^{2-s}\in \vB(\vH)\text{ for all }s\in \R.
\end{align*} Now we use this property with $s=1$ and  \eqref{86a} with $D_r=0$
  although we could  take $D_r$ to be the
contribution to $\tilde \bD B$ from the last term in
\eqref{eq:limit}. In any case \eqref{eq:dif} follows.

We multiply by $f(H)$ from the left and from the right and commute these
factors next to the factor $\tilde \bD B$ in the middle. 
We note 
 \begin{align*}
  [f(H),X_\kappa^{1/4}]= X^{-3/4}O(\kappa^{0}),
 \end{align*} and therefore we can bound the commutation errors 
 obtaining  
\begin{align*}
\begin{split}
  f(H)&X_{\kappa}^{1/4}\parb{\tilde \bD \zeta_\epsilon(B)}X_{\kappa}^{1/4}f(H)\\
&= X_{\kappa}^{1/4}\parb{\theta_\epsilon(B)f(H)(\tilde \bD
  B)f(H)\theta_\epsilon(B)+\theta_\epsilon(-B)f(H)(\tilde \bD B)f(H)
\theta_\epsilon(-B)}X_{\kappa}^{1/4}\\&+
   X^{-3/4}O(\kappa^{0})X^{-3/4}.
\end{split}
\end{align*} Now we can replace $\tilde \bD $ by $\bD $, use
\eqref{eq:limit}  and implement \eqref{eq:mour7d} after commutation of
factors of $r^{-1/2}$ (recall that $f$ was
chosen with small support so that \eqref{eq:mour7d} applies) and then
move the factors of $f(H)$ back where they came from.
As 
in \eqref{eq:lowbnd0} we then  obtain a lower bound of the form
\begin{align*}
 (\sigma-\epsilon^2)f(H)X_\kappa^{1/4}X^{-1/2}\parbb{\eta'_\epsilon(B)+\eta'_\epsilon(-B)}X^{-1/2}X_\kappa^{1/4}f(H) +
   X^{-3/4}O(\kappa^{0})X^{-3/4}
\end{align*} plus 
the
contribution from $K$  that was ignored in  the
heuristic bound \eqref{eq:lowbnd0}. 
 This contribution  is treated by  first fixing  a  big $m\in \N $ such that (with
$\epsilon$ given as before)
\begin{align*}
 \sigma-\epsilon^2-\|K-\chi_{m}K\chi_{m}\|\geq \sigma/2, 
\end{align*}
  and then noting that the  contribution from
 the operator $\chi_{m}K\chi_{m}$   is bounded by $C\|X^{-3/4}\phi\|^2$.

In total we have proved
\begin{align*}
  \|X_\kappa^{-1/4}Y_{\kappa}\phi\|^2=\lim_{n\to
    \infty}\|X_\kappa^{-1/4}Y_{\kappa}f(H)\phi_{n}\|^2\leq C \|X^{-3/4}\phi\|^2,
\end{align*} leading to the desired conclusion,  $\phi\in L_{-1/4}^2$, by
letting $\kappa\to 0$.

 \Step{ Step III} We complete the proof of the assertion 
$\phi\in L^2$. This part  is very similar to the previous part  and therefore we
leave  out the details of the proof. We  consider the observable
    \begin{align*}
      P=P_{\kappa}=f(H) X_\kappa^{1/2}\zeta_\epsilon(B)X_\kappa^{1/2}f(H),\,\kappa>0,
    \end{align*} where $\epsilon>0$ is chosen with $2\epsilon^2<
    \sigma$ and such that \eqref{eq:fundest} applies with $\tilde
    c=\sigma/2$ (rather than $\tilde c=\sigma$ as before).  Redoing
    the commutator arguments with slight modifications and by using
    that $\phi\in L^2_{-1/4}\subseteq  L^2_{-1/2}$ with the proven bound
\begin{align*}
  \|X^{-1/4}\phi\|^2\leq C\|X^{-3/4}\phi\|^2,
\end{align*}
we obtain the bound
\begin{align}\label{eq:fundestsB}
  \|Y_{\kappa}\phi\|^2=\lim_{n\to \infty}\|Y_{\kappa}f(H)\phi_{n}\|^2\leq C\|X^{-3/4}\phi\|^2.
\end{align} 
By letting $\kappa \to 0$ we deduce that $\phi\in L^2$. 
  \end{proof}

 \appendix

\section{Functional calculus}\label{sec:Functional calculus}

  \begin{subequations}
    The following abstract result is a particular time-independent
    version of \cite[Lemma 3.5]{HS} adapted to the proof of Theorem
\ref{thm:priori-decay-b_0}, although the notation for operators does
not conform with that of the proof of Theorem
\ref{thm:priori-decay-b_0}. There are many similar results in the literature, see for
    example \cite[Section 2]{GIS} or \cite[Appendix C]{DG}. The
    function $\eta$ referred to in \ref{it:C} below is the function
    introduced at \eqref{eq:constr}.
\begin{lemma}\label{lem:3.5}  
  Suppose  $B$ is a self-adjoint operator on a complex Hilbert space
  $\vH$, and that $ {H}$ is a symmetric operator on  $\vH$ with
  domain $
  \mathcal{D}=\mathcal{D}{\left(H\right)}=\mathcal{D}{\left(B\right)}$. Suppose that the commutator form $ i{\left[{H},B\right]}$ defined on $ \mathcal{D}$ is a
  symmetric operator with the same  domain $
  \mathcal{D}$. Let $  
\mathbf{D}=\i{\left[{H},\cdot \right]}$
  Then:
  \begin{enumerate}[A)]
  \item \label{it:A}  For any given $ F\in C^{\infty
    }_{\c}{\left({\R}\right)} $ we let $ \tilde{F}\in C^{\infty }_{ \c
    }{\left({\C}\right)}$ denote an almost analytic extension. In
    particular
  \begin{equation}
    \label{82a} 
    F{\left(B\right)} ={\tfrac {1}{\pi }}\int _{{\C}
    }{\left(\bar{\partial} \tilde{F}\right)}{\left(
        z\right)}{\left(B -z\right)}^{-1}dudv,\quad z=u+
\i v,
\end{equation}
and 
\begin{equation}
\label{83a} 
 \mathbf{D}F{\left(B\right)} =-{\tfrac{1}{\pi }}\int _{{\C}
}{\left(\bar{\partial} \tilde{F}\right)}{\left(
    z\right)}{\left(B -z\right)}^{-1}{\left(\mathbf{D}B\right)}{\left(B-z\right)}^{-1}dud v.
\end{equation}

   In particular if 
${\bD} B$ is bounded then for any $\epsilon >0$ (with $
{\left\langle z\right\rangle }={\left(1+|z|^{
      2}\right)}^{{\tfrac{1}{2}} }$) 
\begin{equation}
\label{8883a} 
\left\|\mathbf{D}F{\left(B \right)}\right\|\leq C_{\epsilon
}\sup_{z\in \C}{\Big({\left\langle z\right\rangle }^{%
      \epsilon +2}|\Im z|^{-2}| {\big(\bar{\partial }\tilde{F}
      \big)}{\left(z\right)}| \Big)}
 \|\mathbf{D}B \|.
\end{equation}

 \item \label{it:B} 
Suppose in addition that we can split $ \mathbf{D}B=D+D_{r}$, where $ D$  and
$D_{r} $ are symmetric operators on $
    \mathcal{D}$,  and that similarly for $k=1$ the form $
    \i^{k}\ad^{k}_{B 
}{\left(D\right)} =\i{\big[\i^{k-1}\ad^{k-1}_{B}{\left(D
      \right)},B\big]} $ defined on $ \mathcal{
  D}$ is a symmetric operator on $\mathcal{D}$. Here by definition $\ad^{0}_{B}
{(D)} =D$,  and we note that 
 the form $ \ad^{2}_{B }(D)$  appearing below makes sense without
 further assumptions. Then the contribution from $ D$ to
(\ref{83a}) can be written as

\begin{align}\label{84a} 
-{\tfrac{1}{\pi }}\int
  _{{\C}} &{\big(\bar{\partial} \tilde{F}\big)}{\left(
      z\right)}{\left(B -z\right)}^{-1}D
{\left(B-z\right)}^{
  -1}dudv\nonumber
\\
&={\tfrac{1}{2}}{\left(
    F^{\prime}{\left(B \right)}D+DF^{\prime}{\left(B\right)}\right)}
+R_{1}{\left(t\right)};\\
R_{1}&={\tfrac{1}{
  2\pi }}\int _{\C}{\big( \bar{\partial}\tilde{F}
\big)}{\left(z\right)}{\left(
B-z\right)}^{
-2} \ad^{2}_{B }{\left(D\right)}
{\left(B-z\right)}^{
  -2}dudv.\nonumber
\end{align}
For all  $ f\in C^{\infty }_{\c}{(\R)}$
\begin{align}\label{85a} 
  & {{\tfrac{1}{2}}{\left(f^{ 2}{\left(B\right)}
        D+D f^{2}{\left(B
          \right)}\right)}}\nonumber
  \\
  &=f{\left( B\right)}Df{\left(B \right)}+R_{2} ;
  \\
  R_{2}{\left(t\right)}&=\tfrac{1}{2\pi ^2}\int _{\C}\int
  _{\C}{\big(\bar{\partial} \tilde{f}\big)}{\left(
      z_{2}\right)}{\big(\bar{\partial }\tilde{f}\big)}
  {\left(z_{1}\right)}{\left(
      B-z_{2}\right)}^{-1}{\left(B-z_{1}\right)}^{-1}\nonumber
\\
&\ad^{ 
    2}_{B}{\left( D\right)}{\left(
      B-z_{1}\right)}^{ -1}{\left(B-z_{
        2}\right)}^{-1}du_{1}dv_{ 1}du_{2}dv_{2}.\nonumber
\end{align}

 \item \label{it:C}  Suppose in addition to the previous assumptions that
the operator $ D_{r}$
extends to a bounded self-adjoint operator. Let $F=\eta$ where $\eta$
is the
function from above. Then there exists an almost analytic extension
with
\begin{equation}
\label{87a} 
\big|{\big(\bar{\partial}\tilde{F}\big)}{\left(z\right)}
\big|\leq C_{k}{\left\langle z\right\rangle }^{%
  -1-k}|\Im z|^{k};\ k\in \N,
\end{equation}
yielding the representation 
\begin{equation}
\label{86a} 
 \mathbf{D}F{\left(B\right)} =F^{\prime{\tfrac{1}{2}}}{\left(
    B\right)}DF^{\prime{\tfrac{1}{2}}
}{\left(B\right)} +R_{1}+R_{2}
+R_{3},
\end{equation}
where $ R_{1}$ is given by (\ref{84a}), $ R_{2}$ by (\ref{85a}) with $
f=\sqrt{F^{\prime}} $ and $ R_{3}$ is the contribution from $ D_{r}$
to (\ref{83a}) (the latter possibly estimated as in \eqref{8883a} with $\epsilon=1$).
\end{enumerate}
\end{lemma}

  \end{subequations}


\begin{thebibliography}{DoGa}
   



\bibitem[AH]{AH}   S. Agmon, L. H\"ormander:
\textit{ Asymptotic properties of solutions of differential equations
  with simple characteristics}, J. d'Analyse Math. {\bf 30} (1976),
  1{--}38.

\bibitem[De]{De} J. Derezi\'nski, \emph{Asymptotic completeness for
    $N$-particle long-range quantum systems}, Ann. of Math.
  \textbf{38} (1993), 427--476.


\bibitem[DG]{DG}
J. Derezi{\'n}ski  and C. G{\'e}rard, \emph{Scattering theory of
 classical and quantum {$N$}-particle systems}, Texts and Monographs in
  Physics,  Berlin, Springer 1997.


  \bibitem[FH]{FH}  R. Froese and I. Herbst, \emph{Exponential bounds
and absence of positive eigenvalues for $N$-body Schr{\"o}dinger
operators},   Comm. Math. Phys. {\bf 87} no. 3 (1982/83),  
429--447.
 


\bibitem[GIS]{GIS}  C. G{\'e}rard, H. Isozaki  and E. Skibsted,
  \emph{$N$-body resolvent estimates},
  J. Math. Soc. Japan \textbf{48} no. 1 (1996), 135--160.

\bibitem[Gr]{Gr} G.M. Graf, \emph{Asymptotic completeness for
    $N$-body short-range quantum systems: a new proof},
  Commun. Math. Phys. \textbf{132} (1990), 73--101.



\bibitem[HS]{HS} I. Herbst, E. Skibsted,  \emph{Absence of quantum states corresponding to unstable classical channels},  Ann. Henri Poincar\'e {\bf 9} (2008), 509--552.

\bibitem[H{\"o}]{Ho2} L. H{\"o}rmander, \emph{The analysis of linear
    partial differential operators. {II}-{IV}}, Berlin, Springer
  1983--85.

\bibitem[Is]{Is}
 H. Isozaki, \emph{A generalization of the radiation condition of
 Sommerfeld for $N$-body Schr\"odinger operators}, Duke
 Math. J. \textbf{74}  no. 2  (1994), 557--584.




\bibitem[IS]{IS} K. Ito, E. Skibsted, \emph{Absence of positive
    eigenvalues for hard-core $N$-body systems}, Ann. Inst. Henri Poincar{\'e} {\bf 15} 
  (2014), 2379--2408.




\bibitem[JP]{JP} A. Jensen, P. Perry, \emph{Commutator methods and
    Besov space estimates for {S}chr\"odinger operators}, J. Operator
  Theory \textbf{14} 
   (1985), 181--188.

\bibitem[L1]{L1} W. Littman, \emph{Decay at infinity of solutions to partial differential equations with constant coefficients}, Trans. AMS. {\bf 123}, (1966), 449--459.

\bibitem[L2]{L2} W. Littman, \emph{Decay at infinity of solutions of higher order partial differential equations: removal of the curvature assumption}, Israel J. Math. {\bf 8}, (1970), 403--407.



\bibitem[Pe]{Pe} P. Perry, \emph{Exponential bounds and semifiniteness
    of point spectrum for $N$-body {S}chr\"odinger operators},
  Commun. Math. Phys. \textbf{92} (1984), 481--483.





\bibitem[RS]{RS} M.~Reed, B.~Simon, \emph{Methods of modern
    mathematical physics {I}-{IV}}, New York, Academic Press 1972-78.

\bibitem[Re]{Re} F. Rellich, \emph{\"Uber das asymptotische Verhalten
    der L{\"o}sungen von $\Delta u + k^2u = 0$ in un-
endlichen Gebeiten}, Jber. Deutsch. Math.-Verein. \textbf{53} (1943), 57--65.


\bibitem[Sk]{Sk} E. Skibsted: \emph{Sommerfeld radiation condition at
  threshold}, 
Communication of Partial Differential Equations {\bf 38} (2013), 1601--1625.



  

\bibitem[Ya]{Ya} D. Yafaev, \emph{Eigenfunctions of the continuous spectrum for the N-particle Schr\"odinger operator}, Spectral and scattering theory (Sanda, 1992), 259--286,
Lecture Notes in Pure and Appl. Math., \textbf{161}, Dekker, New York, 1994. 
\end{thebibliography}
\end{document}